\documentclass{birkjour}
 \newtheorem{thm}{Theorem}[section]
 \newtheorem{cor}[thm]{Corollary}
 
 \newtheorem{prop}[thm]{Proposition}
 \theoremstyle{definition}
 \newtheorem{defn}[thm]{Definition}
 \theoremstyle{remark}
 \newtheorem{rem}[thm]{Remark}
 
 \numberwithin{equation}{section}

\def\Z{{\mathbb Z}_2}
\def\ZZ{{\mathbb Z}_2 \otimes {\mathbb Z}_2}
\def\half{\frac{1}{2}}
\def\g{\mathfrak{g}}

\usepackage{bm}
\usepackage{stmaryrd}

\begin{document}

%
%
%
%
%
%
%
%
%

\title[Generalization of Superalgebras to Color Superalgebras]
 {Generalization of Superalgebras to Color Superalgebras and Their Representations}

\author[N. Aizawa]{Naruhiko Aizawa}

\address{%
Department of Physical Science \\
Graduate School of Science \\
Osaka Prefecture University \\
Nakamozu Campus, Sakai, Osaka 599-8531 \\
Japan}

\email{aizawa@p.s.osakafu-u.ac.jp}

\subjclass{Primary 17B75; Secondary 15A66}

\keywords{Color superalgebras, Clifford algebras, Lie superalgebras,  boson-fermion systems.}

\date{January 1, 2004}

\begin{abstract}
For a given Lie superalgebra, two ways of constructing color superalgebras are presented. 
One of them is based on the color superalgebraic nature of the Clifford algebras. 
The method is applicable to any Lie superalgebras and results in color superalgebra of $ {\mathbb Z}_2^{\otimes N} $ grading.  
The other is  discussed with an example, a superalgebra of boson and fermion operators. 
By treating the boson operators as ``second" fermionic sector we obtain a color superalgebra of $\ZZ$ grading. 
A vector field representation of the color superalgebra obtaind from the boson-fermion system is also presented. 
\end{abstract}

\maketitle
\section{Introduction}

We discuss, in this work, two ways of construction of color superalgebras starting from a given Lie superalgebra. 
One of them is a generalization of the method given in \cite{rw2} and the another 
is an application of the author's previous works to a simple example \cite{aktt1,aktt2,AiSe}. 
As a consequence we have a novel families of $\Z^{\otimes N}$ graded and $ \ZZ$ graded  color superalgebras. 

Color superalgebras are a generalization of Lie superalgebras introduced by Rittenberg and Wyler \cite{rw1,rw2} (see also \cite{sch,GrJa}). 
The idea of generalization is to extend the $\Z$ graded structure of the underlying vector space of Lie superalgebra to more general abelian groups 
such as $ \ZZ, {\mathbb Z}_n \otimes {\mathbb Z}_n \otimes {\mathbb Z}_n, $  etc. During the last four decades many works on mathematical aspects of color (super)algebras focused on structure, classification, representations and so on  have been done. See, for example, \cite{sch2,sch3,Sil,ChSiVO,CART,SigSil,AiSe,StoVDJ} (and references therein). 
However, physical applications of the algebraic structures are very limited compared with the fundamental importance of Lie superalgebras in theoretical and mathematical physics. 
First attempt of physical application was made in  \cite{lr}. 
The authors try  to unify the spacetime and internal symmetries (color of quarks) by using $ \mathbb{Z}_3^{\otimes 3} $ graded color superalgebras and  the name ``color (super)algebras" seems to be after this work. 
Then the color (super)algebras appear in the context of de Sitter supergravity \cite{vas}, generalization of quasi-spin formalism \cite{jyw}, string theory \cite{zhe}, generalization of spacetime symmetries \cite{Toro1,Toro2,tol2} and para-statistics \cite{tol}. Despite their mathematical interest, we are not in a position to emphasis physical importance of color (super)algebras. 

Recent investigation of symmetries of differential equations reveals that a $\ZZ$ graded color superalgebra gives  symmetries of physically important equation called 
L\'evy-Leblond equations \cite{aktt1,aktt2}. 
The L\'evy-Leblond equation is a non-relativistic wave equation for a free particle with spin 1/2 \cite{LLE}. It is in a sense ``square root" of Schr\"odinger equation for a free particle and corresponds to a non-relativistic limit of the Dirac equation. 
There are some similarity between the L\'evy-Leblond and Dirac equations: (i) they are formulated by using the Clifford algebra in arbitrary dimensional spacetime, (ii) statistical interpretation of the wavefunctions is possible, (iii) when coupled to static electromagnetic field it gives the gyromagnetic ratio $g=2$ (better than the Pauli equation which gives $ g=1$).  In \cite{aktt1,aktt2} the L\'evy-Leblond equations in one, two and three dimensional space are investigated and the $\ZZ$ graded color superalgebra symmetry is observed for all cases. 

The observation that symmetry of such a fundamental equation of physics is described by a color superalgebra may suggest that color (super)algebras are more ubiquitous in physics. This motivate us to study color (super)algebras more seriously. 
We try to show, in this work, that how color superalgebras are closely related to ordinary Lie superalgebras. More precisely, it is easy to construct  color superalgebras starting from Lie superalgebras. One way presented in this paper is to take a tensor product of a Clifford algebra and a superalgebra.  
The other way is regarding even elements of Lie superalgebra as odd ones. 
This is just happened in the symmetries of L\'evy-Leblond equations \cite{aktt1,aktt2}.

This paper is organized as follows. 
In \S \ref{SEC:DEF} we give a definition of the color (super)algebras of $ \Z^{\otimes N}$ grading. 
We discuss $ Z^{\otimes N}$  graded structures of the Clifford algebras in \S \ref{SEC:CLZN}. 
It is shown that the Clifford  algebra $Cl(p,q)$ exemplifies both  the color superalgebra and the color algebra. 
This fact is used in \S \ref{SEC:LSA2CS} to construct a color superalgebra from a given Lie superalgebra. 
In \S \ref{SEC:BFCS} we provide an another method of converting a Lie superalgebra to a $\ZZ$ graded color superalgebra. 
This is done for a quite simple but physically fundamental Lie superalgebra, i.e., boson-fermion systems.  
In \S \ref{SEC:VecRep}  we discuss a representation of the color superalgebra, denoted by $\mathfrak{bf},$ introduced in \S \ref{SEC:BFCS}. 
It is a straightforward extension of the vector field representation of the Lie algebra. Namely, we consider a space of functions on the color supergroup generated by $\mathfrak{bf}$ and an action of $\mathfrak{bf}$ on the space of function. Then we obtain a realization of $\mathfrak{bf}$ in terms of differential operators. We summarize the results and give some remarks in \S \ref{SEC:CR}.

\section{Definition of color superalgebras} \label{SEC:DEF}

We give a definition of $ \Z^{\otimes N}$ graded ($N$ is a positive integer) color superalgebras and color Lie algebras. The abelian group $ \Z^{\otimes N}$ is the only grading considered in the present work. See \cite{sch} for more general setting. 

Let $\g$ be a vector space over $ \mathbb{C} $ or $ \mathbb{R} $ which is a direct sum of $2^N$ subspaces labelled by an element of the group $ \Z^{\otimes N}:$
\begin{equation}
  \g = \bigoplus_{\bm{\alpha} \in  \Z^{\otimes N}} \g_{\bm{\alpha}}. 
  \label{vecsp}
\end{equation}
We call an element $ \bm{\alpha} = (\alpha_1, \alpha_2, \dots, \alpha_N) \in \Z^{\otimes N} $ of $ \Z^{\otimes N}$ a grading vector and define an inner product of two grading vectors as usual:
\begin{equation}
  \bm{\alpha} \cdot \bm{\beta} = \sum_{i=1}^N \alpha_i \beta_i.
\end{equation}

\begin{defn} \label{DEF:CSA}
 If $\g$ admits a bilinear form $ \llbracket \ , \ \rrbracket : \g \times \g \to \g $ satisfying the following three relations, then $\g$ is called a $ \Z^{\otimes N}$ graded color superalgebra:
  \begin{enumerate}
     \item $ \llbracket \g_{\bm{\alpha}}, \g_{\bm{\beta}} \rrbracket \subseteq \g_{\bm{\alpha}+\bm{\beta}}, $
     \item $ \llbracket X_{\bm{\alpha}}, X_{\bm{\beta}} \rrbracket 
            = -(-1)^{\bm{\alpha}\cdot\bm{\beta}} \,
             \llbracket X_{\bm{\beta}}, X_{\bm{\alpha}} \rrbracket, $ 
     \item $ \llbracket X_{\bm{\alpha}}, \llbracket X_{\bm{\beta}}, X_{\bm{\gamma}} \rrbracket \rrbracket (-1)^{\bm{\alpha}\cdot\bm{\gamma}}
            + \text{cyclic perm.} = 0,
            $
  \end{enumerate}
  where $ X_{\bm{\alpha}} \in \g_{\bm{\alpha}} $ and the third relation is called the graded Jacobi identity. 
\end{defn}
When the inner product $ \bm{\alpha} \cdot \bm{\beta} $ is an even integer the graded Lie bracket $ \llbracket \ , \ \rrbracket $ is understood as a commutator, while it is an anticommutator if  $ \bm{\alpha} \cdot \bm{\beta} $ is an odd integer.  
If $N=1,$ then Definition \ref{DEF:CSA} is identical to the definition of Lie superalgebras so that the color superalgebra is a natural generalization of Lie superalgebra.

If $N=2M$ is an even integer, one may introduce a $ \Z^{\otimes 2M}$ graded color Lie algebra (we call it color algebra, too). We write the components of the grading vector as
\begin{equation}
  \bm{\alpha} = (\alpha_{1,1}, \alpha_{1,2}, \alpha_{2,1},\alpha_{2,2}, \dots,
  \alpha_{M,1}, \alpha_{M,2}) \in \Z^{\otimes 2M}
  \label{grveceven}
\end{equation}
and define an integer $(\bm{\alpha},\bm{\beta})$ by
\begin{equation}
   (\bm{\alpha},\bm{\beta}) = \sum_{k=1}^M \det 
   \begin{pmatrix}
      \alpha_{k,1} & \alpha_{k,2} \\
      \beta_{k,1} & \beta_{k,2}
   \end{pmatrix}.
\end{equation}

\begin{defn}
Let $\g$ be a vector space \eqref{vecsp} with $ N=2M.$ 
If $\g$ admits a bilinear form $ \llbracket \ , \ \rrbracket : \g \times \g \to \g $ satisfying the relations same as in Definition \ref{DEF:CSA} provided that $ \bm{\alpha} \cdot \bm{\beta} $ is replaced with $ (\bm{\alpha},\bm{\beta}), $ 
then $\g$ is called a $ \Z^{\otimes 2M}$ graded color Lie algebra.
\end{defn}
The color algebras are also determined by commutator (for even $ (\bm{\alpha},\bm{\beta})$) and anticommutators (for odd $ (\bm{\alpha},\bm{\beta})$).

%
\section{$\Z^{\otimes N}$ graded structures of Clifford algebras} \label{SEC:CLZN}

The Clifford algebra $Cl(p,q)$ is generated by $ \gamma_i\; (i=1,2,\dots, N=p+q)$ subjected to the relations:
\begin{equation}
  \{ \gamma_i, \gamma_j \} = 2 \eta_{ij}, \qquad 
  \eta = \text{diag}(\underbrace{+1, \dots, +1}_{p}, \underbrace{-1, \dots, -1}_{q})
  \label{CLdefRel}
\end{equation}
The $Cl(p,q)$ is a $2^N$ dimensional unital algebra whose elements are given by a product of the generators:
\[
 1, \ \gamma_i, \ \gamma_i \gamma_j, \ \gamma_i \gamma_j \gamma_k,\ \dots, \ \gamma_1 \gamma_2 \cdots \gamma_N 
\]

 We show that the Clifford algebras give an example of color (super)algebras of $\Z^{\otimes N}$ grading. 
Let $ \bm{\alpha} = (\alpha_1, \alpha_2, \dots, \alpha_N) \in \Z^{\otimes N} $ 
be a grading vector of a color (super)algebra and 
$ \alpha_{i_1}, \alpha_{i_2},\dots, \alpha_{i_r} $ be its non-zero entries. 
We associate an element $\gamma_{\bm{\alpha}} $ of $Cl(p,q)$ with the grading vector $\bm{\alpha}$ by
\begin{equation}
  \gamma_{\bm{\alpha}} = \gamma_{i_1} \gamma_{i_2} \cdots \gamma_{i_r}. 
  \label{association}
\end{equation}
To illustrate the association we give the case of $N=3$ as an example:
\[
   \begin{array}{lclclcl}
    (0,0,0) & & 1 & \qquad & (1,1,0) & & \gamma_1 \gamma_2 \\
    (1,0,0) & & \gamma_1  && (1,0,1) && \gamma_1 \gamma_3 \\
    (0,1,0) & & \gamma_2 && (0,1,1) && \gamma_2 \gamma_3 \\
    (0,0,1) & & \gamma_3 && (1,1,1) && \gamma_1 \gamma_2 \gamma_3
   \end{array}
\]
\begin{prop}
  The association \eqref{association} defines the $\Z^{\otimes N}$ graded color superalgebra structure on $ Cl(p,q) $ for all possible pair  $(p,q).$  
It also defines the $\Z^{\otimes N}$ graded color algebra structure on $Cl(p,q)$  if $ N = p+q$ is an even integer. 
\end{prop}

\begin{proof}
The key of the proof is the relation
\begin{equation}
  \gamma_{\bm{\alpha}} \gamma_{\bm{\beta}} 
  =(-1)^{ \bm{\alpha} \cdot \bm{\beta}+ \sigma(\bm{\alpha}) \sigma(\bm{\beta}) } \gamma_{\bm{\beta}} \gamma_{\bm{\alpha}}, 
  \label{CLcomm}
\end{equation}
where 
\begin{equation}
   \sigma(\bm{\alpha}) = \sum_{k=1}^{N} \alpha_k  \label{sigma-alpha}
\end{equation}  
is the number of non-zero entries of $\bm{\alpha}. $ 
The relation \eqref{CLcomm} is easily verified by noting that the inner product $ \bm{\alpha}\cdot \bm{\beta} $ is equal to the number of common generators in $ \gamma_{\bm{\alpha}} $ and $ \gamma_{\bm{\beta}}. $ 
If $ \gamma_{\bm{\alpha}} = \gamma_{i_1} \gamma_{i_2} \cdots \gamma_{i_r} $ and $ \gamma_{\bm{\beta}} = \gamma_{j_1} \gamma_{j_2} \cdots \gamma_{j_s} $ share $m$ common generators $ \gamma_{\ell_1}, \gamma_{\ell_2}, \dots, \gamma_{\ell_m}$, then the relative sign of $ \gamma_{\bm{\alpha}}\gamma_{\bm{\beta}} $ and $ \gamma_{\bm{\beta}} \gamma_{\bm{\alpha}} $ is $ (-1)^{rs-m} $ which gives the relation \eqref{CLcomm}.  

 It is also obvious from \eqref{CLdefRel} and $\Z^{\otimes N}$ nature of the grading vector that we have the relation
\begin{equation}
   \gamma_{\bm{\alpha}} \gamma_{\bm{\beta}} = \kappa(\bm{\alpha},\bm{\beta}) \gamma_{\bm{\alpha}+\bm{\beta}}, 
   \label{Clproduct}
\end{equation}
where $ \kappa(\bm{\alpha},\bm{\beta}) = \pm 1 $ determined by $\bm{\alpha},\bm{\beta} $ and $ p, q$ (for the sake of simplicity of notations $p,q$ dependence is not indicated).  

  Firstly, we prove the color superalgebra structure of $ Cl(p,q). $ 
Closure of the $\Z^{\otimes N}$ graded Lie bracket is easily seen from \eqref{CLcomm} and \eqref{Clproduct}:
\begin{align}
  \llbracket \gamma_{\bm{\alpha}}, \gamma_{\bm{\beta}} \rrbracket 
  &= \gamma_{\bm{\alpha}} \gamma_{\bm{\beta}} - (-1)^{\bm{\alpha}\cdot\bm{\beta}} \gamma_{\bm{\beta}} \gamma_{\bm{\alpha}} 
  \nonumber \\
  &= (1-(-1)^{\sigma(\bm{\alpha}) \sigma(\bm{\beta})} ) \gamma_{\bm{\alpha}} \gamma_{\bm{\beta}}
  \nonumber \\
  &= (1-(-1)^{\sigma(\bm{\alpha}) \sigma(\bm{\beta})} ) \kappa(\bm{\alpha},\bm{\beta}) \gamma_{\bm{\alpha}+\bm{\beta}}. 
  \label{CLclosure}
\end{align}
The graded Jacobi identities are also verified by the straightforward computation. Using \eqref{CLcomm} and the second equality in \eqref{CLclosure}   one may see the identities
\begin{align*}
   \llbracket \gamma_{\bm{\alpha}}, \llbracket \gamma_{\bm{\beta}}, \gamma_{\bm{\gamma}} \rrbracket \rrbracket 
   (-1)^{\bm{\alpha} \cdot \bm{\gamma}}
   &= f(\bm{\alpha}, \bm{\beta},\bm{\gamma}) 
   (-1)^{\bm{\alpha} \cdot \bm{\gamma}} \gamma_{\bm{\alpha}} 
   \gamma_{\bm{\beta}} \gamma_{\bm{\gamma}},
    \\   
   \llbracket \gamma_{\bm{\beta}}, \llbracket \gamma_{\bm{\gamma}}, \gamma_{\bm{\alpha}} \rrbracket \rrbracket 
   (-1)^{\bm{\beta} \cdot \bm{\alpha}}
   &= f(\bm{\beta}, \bm{\gamma},\bm{\alpha}) 
   (-1)^{\bm{\alpha} \cdot \bm{\gamma} + \sigma(\bm{\alpha}) \sigma(\bm{\beta}+\bm{\gamma}) }
   \gamma_{\bm{\alpha}} \gamma_{\bm{\beta}} \gamma_{\bm{\gamma}},
   \\
   \llbracket \gamma_{\bm{\gamma}}, \llbracket \gamma_{\bm{\alpha}}, \gamma_{\bm{\beta}} \rrbracket \rrbracket 
   (-1)^{\bm{\gamma} \cdot \bm{\beta}}
   &= f(\bm{\gamma}, \bm{\alpha},\bm{\beta}) 
   (-1)^{\bm{\alpha} \cdot \bm{\gamma} + \sigma(\bm{\gamma}) \sigma(\bm{\alpha}+\bm{\beta}) }
   \gamma_{\bm{\alpha}} \gamma_{\bm{\beta}} \gamma_{\bm{\gamma}},
\end{align*}
with
\[
   f(\bm{\alpha}, \bm{\beta},\bm{\gamma}) = 
   (1-(-1)^{\sigma(\bm{\alpha}) \sigma(\bm{\beta}+\bm{\gamma})} )
  (1-(-1)^{\sigma(\bm{\beta}) \sigma(\bm{\gamma})}).
\]
Sum of them vanishes due to 
$ \sigma(\bm{\alpha}+\bm{\beta}) = \sigma(\bm{\alpha}) + \sigma(\bm{\beta}) \ \pmod 2. $ 
This completes the proof of a color superalgebra structure of $Cl(p,q).$ 

Secondly, we prove the color algebra structure of $Cl(p,q)$ for even $N=2M.$ 
In this case the components of the grading vector is denoted as in \eqref{grveceven}. We introduce some more notations:
\begin{align*}
   \bm{\alpha}_k &= (\alpha_{k,1}, \alpha_{k,2}), 
   \\
   \sigma(\bm{\alpha}_k) &= \alpha_{k,1} + \alpha_{k,2},
   \\
   \bm{\alpha}_k \cdot \bm{\beta}_k &= \alpha_{k,1} \beta_{k,1} + \alpha_{k,2} \beta_{k,2},
   \\
   \bm{\alpha}_k \times \bm{\beta}_k &= \alpha_{k,1} \beta_{k,2} - \alpha_{k,2} \beta_{k,1}.
\end{align*}
Then it is not difficult to see the identity
\begin{equation}
   \bm{\alpha}_k \times \bm{\beta}_k =  \bm{\alpha}_k \cdot \bm{\beta}_k 
   + \sigma(\bm{\alpha}_k) \sigma(\bm{\beta}_k) \pmod 2
\end{equation}
It follows that
\begin{equation}
   (\bm{\alpha},\bm{\beta}) = \bm{\alpha} \cdot \bm{\beta} 
   + \sum_{k=1}^M \sigma(\bm{\alpha}_k) \sigma(\bm{\beta}_k) \pmod 2
   \label{maprel}
\end{equation}
Together with \eqref{Clproduct} one may verify the closure of the graded Lie bracket:
\begin{align}
  \llbracket \gamma_{\bm{\alpha}}, \gamma_{\bm{\beta}} \rrbracket 
  &= \gamma_{\bm{\alpha}} \gamma_{\bm{\beta}} - (-1)^{(\bm{\alpha},\bm{\beta})} \gamma_{\bm{\beta}} \gamma_{\bm{\alpha}} 
  \nonumber \\
  &= (1-(-1)^{\sigma(\bm{\alpha}) \sigma(\bm{\beta}) + \sum_{k=1}^M \sigma(\bm{\alpha}_k) \sigma(\bm{\beta}_k)} ) \gamma_{\bm{\alpha}} \gamma_{\bm{\beta}}
  \nonumber \\
  &= (1-(-1)^{\sigma(\bm{\alpha}) \sigma(\bm{\beta})+ \sum_{k=1}^M \sigma(\bm{\alpha}_k) \sigma(\bm{\beta}_k)} ) \kappa(\bm{\alpha},\bm{\beta}) \gamma_{\bm{\alpha}+\bm{\beta}}. 
  \label{CLclosure2}
\end{align}
To verify the graded Jacobi identities we note the following relations which are  obtained using \eqref{CLcomm} and \eqref{CLclosure2}:
\begin{align*}
   \llbracket \gamma_{\bm{\alpha}}, \llbracket \gamma_{\bm{\beta}}, \gamma_{\bm{\gamma}} \rrbracket \rrbracket 
   (-1)^{(\bm{\alpha}, \bm{\gamma})}
   &= h(\bm{\alpha}, \bm{\beta},\bm{\gamma}) 
   (-1)^{\bm{\alpha} \cdot \bm{\gamma} + \sum_k \sigma(\bm{\alpha}_k) \sigma(\bm{\gamma}_k)} \, \gamma_{\bm{\alpha}} 
   \gamma_{\bm{\beta}} \gamma_{\bm{\gamma}},
    \\   
   \llbracket \gamma_{\bm{\beta}}, \llbracket \gamma_{\bm{\gamma}}, \gamma_{\bm{\alpha}} \rrbracket \rrbracket 
   (-1)^{(\bm{\beta}, \bm{\alpha})}
   &= h(\bm{\beta}, \bm{\gamma},\bm{\alpha}) 
   (-1)^{\bm{\alpha} \cdot \bm{\gamma} + \sigma(\bm{\alpha}) \sigma(\bm{\beta}+\bm{\gamma}) + \sum_k \sigma(\bm{\beta}_k) \sigma(\bm{\alpha}_k) }
   \\ 
   & 
    \times \gamma_{\bm{\alpha}} \gamma_{\bm{\beta}} \gamma_{\bm{\gamma}},
   \\
   \llbracket \gamma_{\bm{\gamma}}, \llbracket \gamma_{\bm{\alpha}}, \gamma_{\bm{\beta}} \rrbracket \rrbracket 
   (-1)^{(\bm{\gamma}, \bm{\beta})}
   &= h(\bm{\gamma}, \bm{\alpha},\bm{\beta}) 
   (-1)^{\bm{\alpha} \cdot \bm{\gamma} + \sigma(\bm{\gamma}) \sigma(\bm{\alpha}+\bm{\beta}) + \sum_k \sigma(\bm{\gamma}_k) \sigma(\bm{\beta}_k)}
   \\
   & \times 
   \gamma_{\bm{\alpha}} \gamma_{\bm{\beta}} \gamma_{\bm{\gamma}},
\end{align*}
with
\begin{align*}
   h(\bm{\alpha}, \bm{\beta},\bm{\gamma}) &= 
   (1-(-1)^{\sigma(\bm{\alpha}) \sigma(\bm{\beta}+\bm{\gamma})
   + \sum_k \sigma(\bm{\alpha}_k) \sigma(\bm{\beta}_k+\bm{\gamma}_k)} )
   \\
  & \times 
  (1-(-1)^{\sigma(\bm{\beta}) \sigma(\bm{\gamma}) + \sum_k \sigma(\bm{\beta}_k) \sigma(\bm{\gamma}_k)}).
\end{align*}
One may easily see that sum of the relations vanishes identically. 
This completes the proof of a  $\Z^{\otimes 2M}$ color algebra structure of $Cl(p,q).$ 
\end{proof}

\section{From a Lie superalgebra to a $\Z^{\otimes N}$ graded color superalgebra} 
\label{SEC:LSA2CS}

The authors of  \cite{rw2} give a method constructing a $\ZZ$ graded color superalgebra from a given Lie superalgebra (see also \cite{rw1,lr}). 
They consider a tensor product of Pauli matrices and a superalgebra and 
as a consequence dimension of the obtained $\ZZ$ graded color superalgebra is double of the original Lie superalgebra. 

 In this section we generalize the construction of \cite{rw2} to a color superalgebra of $\Z^{\otimes N}$ grading. This is done with the $\Z^{\otimes N}$ graded structure of $Cl(p,q)$ discussed in \S \ref{SEC:CLZN}. 
 Let us consider a Lie superalgebra $\g$ defined by the relations:
\begin{align}
   [T_m, T_n] &= C_{mn}^{k} T_k,  \qquad
   [T_m, Q_{\mu} ] = C_{m\mu}^{\nu} Q_{\nu},
   \nonumber \\
   \{ Q_{\mu}, Q_{\nu} \} &= C_{\mu \nu}^{m} T_m. 
   \label{SLAeg}
\end{align}
Let $ \bm{\alpha} \in \Z^{\otimes N} $ be a grading vector and $ \sigma(\bm{\alpha}) $ be an integer defined in \eqref{sigma-alpha}. 
Further let $ \gamma_{\bm{\alpha}}  $ be an element of $ Cl(p,q) $ defined by \eqref{association}. 
A tensor product of $\g$ and $Cl(p,q)$ gives a  $\Z^{\otimes N}$ color superalgebra.

\begin{thm} \label{THM1}
We set
\begin{align}
   X_{\bm{\alpha},m} &= \gamma_{\bm{\alpha}} \otimes T_m,
   \quad (\sigma(\bm{\alpha}) \ \text{even})
   \nonumber \\
   X_{\bm{\alpha},\mu} &= \gamma_{\bm{\alpha}} \otimes Q_{\mu}, 
   \quad (\sigma(\bm{\alpha}) \ \text{odd})
   \label{Z2toZN}
\end{align}
then $ X_{\bm{\alpha},m}, X_{\bm{\alpha},\mu} $ are a basis of a $\Z^{\otimes N}$ color superalgebra. The defining relations of the color superalgebra are given with the structure constants same as the ones of $\g$:
\begin{align}
   [ X_{\bm{\alpha}, a}, X_{\bm{\beta}, b} ] &= \kappa(\bm{\alpha}, \bm{\beta}) C_{ab}^c X_{\bm{\alpha}+\bm{\beta},c},
   \qquad (\bm{\alpha}\cdot\bm{\beta} \ \text{even}),
   \nonumber \\
   \{ X_{\bm{\alpha}, a}, X_{\bm{\beta}, b} \} &= \kappa(\bm{\alpha}, \bm{\beta}) C_{ab}^c X_{\bm{\alpha}+\bm{\beta},c},
   \qquad (\bm{\alpha}\cdot\bm{\beta} \ \text{odd}).
   \label{ZNextendedCom}
\end{align} 
The dimension of the color super algebra is $ 2^{N-1}$ times of $\dim \g.$  
\end{thm}

\begin{proof}
 In order to simplify the computation we introduce a new notation for the Lie superalgebra $\g.$ Its element is denoted by $ \mathcal{T}_a$ and  the parity of $ \mathcal{T}_a $ by $ |a|. $ The $ \mathbb{Z}_2 $ graded Lie bracket is denoted by $ [\ , \ ]_{\pm},$ then the relations in \eqref{SLAeg} are written as
\begin{equation}
    [ \mathcal{T}_a, \mathcal{T}_b ]_{\pm} = C_{ab}^c \mathcal{T}_c. 
    \label{Z2comm}
\end{equation}
Note that, for an element $ X_{\bm{\alpha},a} = \gamma_{\bm{\alpha}} \otimes \mathcal{T}_a, $  
the parity of $ \sigma(\bm{\alpha}) $  is equal to the parity of $ \mathcal{T}_a. $  
Therefore the $\Z^{\otimes N}$ graded bracket for $ X_{\bm{\alpha},a}$ is converted to the $ \mathbb{Z}_2 $ graded bracket for $\mathcal{T}_a:$
\begin{align}
   \llbracket X_{\bm{\alpha},a}, X_{\bm{\beta},b} \rrbracket 
   &= 
   \gamma_{\bm{\alpha}} \gamma_{\bm{\beta}} \otimes \mathcal{T}_a \mathcal{T}_b
   - (-1)^{\bm{\alpha} \cdot \bm{\beta}} \gamma_{\bm{\beta}} \gamma_{\bm{\alpha}} \otimes \mathcal{T}_b \mathcal{T}_a
   \nonumber \\
   &\overset{\eqref{CLcomm}}{=} 
   \gamma_{\bm{\alpha}} \gamma_{\bm{\beta}} \otimes 
    (\mathcal{T}_a \mathcal{T}_b   - (-1)^{|a||b|} \mathcal{T}_b \mathcal{T}_a)
   \nonumber  \\
   &\overset{\eqref{Z2comm}}{=}
   \gamma_{\bm{\alpha}} \gamma_{\bm{\beta}} \otimes 
    C_{ab}^c \mathcal{T}_c
   \nonumber \\
   &\overset{\eqref{Clproduct}}{=} 
   \kappa(\bm{\alpha}, \bm{\beta})  C_{ab}^c X_{\bm{\alpha+\beta},c}.
   \label{ZNCommu}
\end{align}
This is the relations identical to the ones in  \eqref{ZNextendedCom}. 

Due to the similar computation, the $\Z^{\otimes N}$ graded Jacobi identities are reduced to the $ \mathbb{Z}_2 $ graded ones. 
Using the second equality of \eqref{ZNCommu} one may verifies the relation:
\begin{align*}
  &\llbracket X_{\bm{\alpha},a}, \llbracket X_{\bm{\beta},b}, X_{\bm{\gamma},c} \rrbracket \rrbracket
  \\
  &= \gamma_{\bm{\alpha}} \gamma_{\bm{\beta}} \gamma_{\bm{\gamma}} \otimes
  \mathcal{T}_a [ \mathcal{T}_b, \mathcal{T}_c ]_{\pm} 
  - (-1)^{ \bm{\alpha} \cdot (\bm{\beta}+\bm{\gamma}) }
  \gamma_{\bm{\beta}} \gamma_{\bm{\gamma}} \gamma_{\bm{\alpha}} \otimes 
  [\mathcal{T}_b, \mathcal{T}_c]_{\pm} \mathcal{T}_a
  \\
  &\overset{\eqref{CLcomm}}{=} \gamma_{\bm{\alpha}} \gamma_{\bm{\beta}} \gamma_{\bm{\gamma}} \otimes
  [ \mathcal{T}_a, [\mathcal{T}_b, \mathcal{T}_c]_{\pm} ]_{\pm}.
\end{align*}
It follows that
\begin{align*}
  & \llbracket X_{\bm{\alpha},a}, \llbracket X_{\bm{\beta},b}, X_{\bm{\gamma},c} \rrbracket \rrbracket (-1)^{\bm{\alpha} \cdot \bm{\beta}} 
   + \text{cyclic perm.}
   \\
   &= (-1)^{\bm{\alpha} \cdot \bm{\gamma} + |a||c|} 
   \gamma_{\bm{\alpha}} \gamma_{\bm{\beta}} \gamma_{\bm{\gamma}} \otimes 
   ([ \mathcal{T}_a, [\mathcal{T}_b, \mathcal{T}_c]_{\pm} ]_{\pm} (-1)^{|a| |c|}
   + \text{cyclic perm.})
   \\
   &= 0. 
\end{align*}
The last equality is due to the $ \mathbb{Z}_2 $ graded Jacobi identity of $\g.$ 
\end{proof}

Theorem \ref{THM1} shows that any representation of $\g$ is lifted up to a representation of the $\Z^{\otimes N}$ graded color superalgebra.  
The construction \eqref{Z2toZN} implies that one may reverse the procedure. 
It is possible to construct a Lie superalgebra from any $\Z^{\otimes N}$ graded color superalgebra.
\begin{cor}
Let $ X_{\bm{\alpha},a}$ be a basis of a $\Z^{\otimes N}$ graded color superalgebra. 
Then $ T_{\bm{\alpha},a} = \gamma_{\bm{\alpha}} \otimes X_{\bm{\alpha},a} $ is a basis of Lie superalgebra. 
The parity of $ T_{\bm{\alpha},a} $ is equal to the parity of $ \sigma(\bm{\alpha}). $ 
\end{cor}
Proof of the Corollary is easy so we omit it. 
Of course, the Lie superalgebra generated by $T_{\bm{\alpha},a} $ is not isomorphic the Lie superalgebra $\g$ defined in \eqref{SLAeg}.

\section{$\ZZ$ generalization of boson-fermion systems} \label{SEC:BFCS}

In this section we discuss a construction of $\ZZ$ graded color superalgebra 
from a given Lie superalgebra. The construction is different from the one discussed in \S \ref{SEC:LSA2CS} so that the resulted color superalgebras have, in general, different dimension from the one in Theorem \ref{THM1}. 
The simplest case of the examples given  in this section has a subalgebra which is isomorphic to  the color superalgebra given in the first paper by Rittenberg and Wyler \cite{rw1}.

The simplest Lie superalgebra would be a system consisting of independent bosons and fermions. Let $ a_i, a_i^{\dagger} $ be boson annihilation, creation operators and $ \alpha_i, \alpha_i^{\dagger}$ be the fermionic ones. 
We also introduce an idempotent operator $F$ anticommuting with fermions. 
The boson-fermion superalgebra is a $\Z$ graded algebra with the assignment of degree 
\begin{equation}
    (0) \;:\: 1,\ a_i, \ a_i^{\dagger}
    \qquad
    (1) \;:\; \alpha_i, \ \alpha_i^{\dagger}, \ F
\end{equation}
and defined by the relations:
\begin{align}
         &  [a_i, a_j^{\dagger}] = \delta_{i,j}, \qquad \{ \alpha_i, \alpha_j \} = \delta_{i,j}, 
         \qquad \{F, F\} = 2,  \label{BFrel1} \\
         &  [F, a_i] = [F, a_i^{\dagger}] = \{ F, \alpha_i \} = \{ F, \alpha_i^{\dagger} \} = 0. \label{BFrel2}
\end{align}

One should not underestimate the importance of this superalgebra. 
The boson-fermion systems are found in physics literatures on various fields 
and they are basic building blocks of representations of Lie superalgebras. 

We now construct a color superalgebra with $\ZZ$ graded structure 
by the idea same as \cite{aktt1,aktt2,AiSe}, that is, by changing the grading of the boson-fermion systems. We treat the bosons  as ``fermionic" elements and consider the following assignment of $\ZZ$ degree
\begin{equation}
  \begin{array}{ccl}
        (0,0) &: & 1, \ A_{ij}, \ A_{ij}^{\dagger}, \ N_{ij} \\
        (1,0) &: & \alpha_i, \ \alpha_i^{\dagger},\ \beta_i, \ \beta_i^{\dagger} \\
        (0,1) &: & a_i, \ a_i^{\dagger} \\
        (1,1) &: & F  
  \end{array} 
  \label{BFZ2elements}
\end{equation}
The $\ZZ$ degree of $ a_i, a_i^{\dagger} $ and $ F$ forces us to introduce new elements, at least the following ones:
\begin{align}
  A_{ij} &= \half \{a_i, a_j \}, \quad A_{ij}^{\dagger} = \half \{ a_i^{\dagger}, a_j^{\dagger} \}, \quad 
  N_{ij} = \half \{a_i^{\dagger}, a_j \},
  \\
  \beta_{i} &= \half \{ a_i, F\},\quad \beta_{i}^{\dagger} = \half \{ a_i^{\dagger}, F\}.
\end{align}
From the viewpoint of the original $\mathbb{Z}_2$ graded  boson-fermion systems, now we are working on the universal enveloping algebra of the boson-fermion superalgebra. 
With the aid of the relations (\ref{BFrel1}) and (\ref{BFrel2}) one may verify the relations given below which respect the $\ZZ$ grading. It is an interesting observation that the elements given in (\ref{BFZ2elements}) close in the $\ZZ$ (anti)commutation relations without introducing any more higher order elements.  

\medskip
\noindent
$ (0,0)$-$(0,0)$ sector:
\begin{align}
 [A_{ij}, A_{k\ell}^{\dagger}] &= \delta_{jk} N_{\ell i} + \delta_{i\ell} N_{kj} + \delta_{ik} N_{\ell j} + \delta_{j\ell} N_{ki},
 \label{BFdef1}
 \\
 [A_{ij}, N_{k\ell}] &= \delta_{ik} A_{j\ell} + \delta_{jk} A_{i\ell},
 \\
 [A_{ij}^{\dagger}, N_{k\ell}] &= -\delta_{i\ell} A_{kj}^{\dagger} - \delta_{j\ell} A_{ik}^{\dagger},
 \\
 [N_{ij}, N_{k\ell}] &= \delta_{jk} N_{i\ell} - \delta_{i\ell} N_{kj}. 
\end{align}

\noindent
$ (0,0)$-$(1,0)$ sector:
\begin{align}
  [A_{ij}, \beta_k^{\dagger}] &= \delta_{ik} \beta_j + \delta_{jk} \beta_i,
  \\
  [A_{ij}^{\dagger}, \beta_k] &= - \delta_{ik} \beta_j^{\dagger}-\delta_{jk}\beta_i^{\dagger},
  \\
  [N_{ij}, \beta_k] &= -\delta_{ik} \beta_j,
  \\
  [N_{ij}, \beta_k^{\dagger}] &= \delta_{jk}\beta_i^{\dagger}.
\end{align}

\noindent
$ (0,0)$-$(0,1)$ sector:
\begin{align}
  [A_{ij}, a_k^{\dagger}] &= \delta_{ik} a_j + \delta_{jk} a_i,
  \\
  [A_{ij}^{\dagger}, a_k] &= -\delta_{ik} a_j^{\dagger} - \delta_{jk} a_i^{\dagger},
  \\
  [N_{ij}, a_k] &= -\delta_{ik} a_j,
  \\
  [N_{ij}, a_k^{\dagger}] &= \delta_{jk} a_i^{\dagger}.
\end{align}

\noindent
$ (1,0)$-$(1,0)$ sector:
\begin{align}
  \{ \alpha_i, \alpha_j^{\dagger} \} &= \delta_{ij}, 
  &\{ \beta_i, \beta_j \} = 2 A_{ij},
  \\
  \{ \beta_i, \beta_j^{\dagger} \} &= 2 N_{ji}, 
 & \{ \beta_i^{\dagger}, \beta_j^{\dagger} \} = 2 A_{ij}^{\dagger}.
\end{align}

\noindent
$ (1,0)$-$(0,1)$ sector:
\begin{equation}
   [\beta_i, a_j^{\dagger}] = \delta_{ij} F, 
   \qquad \qquad
   [\beta_i^{\dagger}, a_j] = -\delta_{ij} F.
\end{equation}

\noindent
$ (1,0)$-$(1,1)$ sector:
\begin{equation}
  \{ \beta_i, F \} = 2 a_i, \qquad \qquad 
  \{ \beta_i^{\dagger}, F\} = 2 a_i^{\dagger}
\end{equation}

\noindent
$ (0,1)$-$(0,1)$ sector:
\begin{equation}
   \{ a_i, a_j \} = 2A_{ij}, \qquad 
   \{ a_i, a_j^{\dagger} \} = 2 N_{ji},
   \qquad
   \{ a_i^{\dagger}, a_j^{\dagger} \} = 2A_{ij}^{\dagger}.
\end{equation}

\noindent
$ (0,1)$-$(1,1)$ sector:
\begin{equation}
  \{ a_i, F \} = 2\beta_i, \qquad \qquad
  \{ a_i^{\dagger}, F \} = 2 \beta_i^{\dagger}.
  \label{BFdef2}
\end{equation}
There exist no nonvanishing relations in $(0,0)$-$(1,1)$ sector. 

 One may verify by the straightforward computation that the  $\ZZ$ relations obtained above satisfy the graded Jacobi identity of the color superalgebra.  
We thus obtain a new color superalgebra with $\ZZ$ grading.
\begin{prop}
 The elements in \eqref{BFZ2elements} with the relations \eqref{BFdef1}-\eqref{BFdef2} defines a color superalgebra with $\ZZ$ grading structure. 
 Any representations of the boson-fermion system are able to promote to the 
 representations of the color superalgebra. 
\end{prop}

\begin{rem}
 The boson commutation relations of $a_i, a_i^{\dagger}$  never hold true in the color superalgebra. $F$ is no longer idempotent in the color superalgebra, either.  
\end{rem}

\begin{rem} \label{REM:three}
 The fermion operators $\alpha_i, \alpha_i^{\dagger} $ are decoupled from the color superalgebra. This means that a color superalgebra is obtained from the boson subalgebra. The color superalgebra obtained from the single boson and $F$ is identical to the one in \cite{rw1} (see eq.(2.21)).
\end{rem}

%
\section{Representations of the color superalgebra obtained from single-mode boson-fermion system} \label{SEC:VecRep}

We consider a vector field representations of the color superalgebra obtained in \S \ref{SEC:BFCS}. The purpose of this is to show that the color superalgebra has non-trivial representations beside the realization in the enveloping algebra of the boson-fermion systems. 
We concentrate on the simplest case, i.e., the color superalgebra obtained from the single-mode boson-fermion system.  In our convention the color superalgebra we investigate has the following basis:
\begin{align}
   & (0,0) \ : \ 1, \ A,\ A^{\dagger},\ N & & (1,0) \ : \alpha, \ \alpha^{\dagger}, \ \beta, \ \beta^{\dagger}
   \nonumber \\
   & (0,1) \ : \ a, \ a^{\dagger} & & (1,1) \ : \ F
   \label{BasisBF}
\end{align}
We denote the color superalgebra by $ \mathfrak{bf}. $ 
We decompose $ \mathfrak{bf} $ as follows:
\begin{align}
    \mathfrak{bf}_+ \quad &: \quad A^{\dagger}, \ a^{\dagger}, \ \alpha^{\dagger}, \ \beta^{\dagger}
    \nonumber \\
    \mathfrak{bf}_0 \quad &: \quad 1, \ N,\ F
    \nonumber \\
    \mathfrak{bf}_- \quad &: \quad A, \ a, \ \alpha, \ \beta
    \label{TriDeco}
\end{align}
This is a kind of triangular decomposition. Indeed, it corresponds to a triangular decomposition of the subalgebra mentioned in Remark \ref{REM:three}. It was done according to the eigenvalue of ad$ N.$ The $\mathfrak{bf} $ is a direct sum of the subalgebra and the fermion algebra generated by $ \alpha, \alpha^{\dagger}.$ Furthermore, $\mathfrak{bf}_0$ is hermitian and $\mathfrak{bf}_{\pm} $ are hermitian conjugate each other.  Thus \eqref{TriDeco} would be a natural decomposition of $\mathfrak{bf}. $ 

The color superalgebra $\mathfrak{bf}$ may generate a color supergroup by exponential mapping. Parameters of the group are given by a $ \ZZ $ extension of the Grassmann numbers defined in \cite{rw2}. The $ \ZZ $ extended Grassmann numbers $ \zeta_{\bm{\alpha}, i} \ (\bm{\alpha} \in \ZZ) $ are defined by the relations:
\begin{equation}
  \llbracket \zeta_{\bm{\alpha}, i}, \zeta_{\bm{\beta}, j}  \rrbracket = 0.
\end{equation}
As an corollary of Theorem \ref{THM1}, $ \zeta_{\bm{\alpha}, i} $ is realized in terms of the Grassmann number $\xi_{\mu} $ and the Clifford algebra $ Cl(p,q),\; p+q = 2:$
\begin{align*}
  \zeta_{(0,0),m} &= 1 \otimes x_m, \quad 
  \zeta_{(1,0),\mu} = \gamma_1 \otimes \xi_{\mu},
  \\
  \zeta_{(0,1),\mu} &= \gamma_2 \otimes \xi_{\mu},  \quad
  \zeta_{(1,1),m} = \gamma_1 \gamma_2 \otimes x_m,
\end{align*}
where $ x_m \in \mathbb{R}. $

 An element of the color supergroup is given as
\begin{equation}
  g_+ = \exp( x A^{\dagger}) \exp(\psi a^{\dagger}) \exp(\theta_1 \beta^{\dagger}) 
  \exp( \theta_2 \alpha^{\dagger}). 
  \label{groupelements}
\end{equation}
Here we have changed the notation for the $\ZZ$ extended Grassmann numbers form $ \zeta_{\bm{\alpha}} $ to more readable ones. The $ \ZZ $ degree of $x, \psi, \theta_i$ are obvious from \eqref{groupelements}. The element $ g_+$ is generated by $ \mathfrak{bf}_+. $ 
The elements $ g_-$ and $ g_0, $ which are generated by $ \mathfrak{bf}_-$ and $\mathfrak{bf}_0,$ are given in a similar way. 

Let us consider a space of functions on the color supergroup with the special property (see eg. \cite{Dob}):
\begin{equation}
     f(g_+ g_0 g_-) = f(g_+).  \label{RightCov}
\end{equation}
Namely, $f$ is a function on the coset $G/G_0G_-$ of the color supergroup.  
This is for the sake of simplicity of the final result. It is, in general, not necessary to consider such a special property. 
Without \eqref{RightCov} the final result contains more parameters. 

Derivative and integral for the $\ZZ$ extended Grassmann number are defined in \cite{AiSe}. The definition is same as the left derivative of Grassmann numbers. For example,
\[
  \frac{\partial}{\partial \theta_2} x \theta_1 \theta_2 = - x \theta_1, 
  \qquad
  \frac{\partial}{\partial \psi} \theta_1 \psi = \theta_1. 
\]
Now we are ready to define an action of $ \mathfrak{bf}$ on the space of functions. 
The definition is the same as the standard Lie theory \cite{AiSe}. 

\begin{defn}
  Let $ Y \in \mathfrak{bf}$ and $ \tau $ be a $\ZZ $ extended Grassmann number of the $\ZZ$ degree same as $Y.$ The left action of $Y$ on $ f(g_+) $ is defined by
  \begin{equation}
     Y f(g_+) = \left. \frac{d}{d \tau} f( e^{-\tau Y} g_+) \right|_{\tau = 0}
  \end{equation}
\end{defn}
It is not difficult to verify that the left action gives a realization of $ \mathfrak{bf}$ on the space of functions $f.$ 

\begin{prop}
 The vector field representation of $ \mathfrak{bf} $ is given as follows:
 \begin{align*}
    A^{\dagger} &= - \frac{\partial}{\partial x},
    \\
    a^{\dagger} &= -\frac{\partial}{\partial \psi} + \psi \frac{\partial}{\partial x},
    \\
    \alpha^{\dagger} &= - \frac{\partial}{\partial \theta_2},
    \\
    \end{align*}
    \begin{align*}
    \beta^{\dagger} &= -\frac{\partial}{\partial \theta_1} + \theta_1 \frac{\partial}{\partial x},
    \\
    N &= - 2x \frac{\partial}{\partial x} - \psi \frac{\partial}{\partial \psi} - \theta_1 \frac{\partial}{\partial\theta_1},
    \\
    F &= 2\theta_1 \frac{\partial}{\partial \psi} + 2 \psi \frac{\partial}{\partial \theta_1},
    \\
    A &= -4x\left( x \frac{\partial}{\partial x} + \psi \frac{\partial}{\partial \psi} + \theta_1 \frac{\partial}{\partial \theta_1} \right),
    \\
    a &= -2 x \frac{\partial}{\partial \psi} + 2 x \psi \frac{\partial}{\partial x} + 2 \psi \theta_1 \frac{\partial}{\partial \theta_1}, 
    \\
    \alpha &= - \theta_2,
    \\
    \beta &= 2x \theta_1 \frac{\partial}{\partial x} + 2 \psi \theta_1 \frac{\partial}{\partial \psi} - 2 x \frac{\partial}{\partial \theta_1}.
 \end{align*}
\end{prop}

Generalizing the vector field representation to the multi-mode systems may not be difficult. 

%
\section{Concluding remarks} \label{SEC:CR}

An interesting observation of the present work is that the Clifford  algebras exemplify both color superalgebras and color algebras. 
This allows us to convert any Lie superalgebra to  a color  superalgebra and  vice versa. 
Thus any representation of the original Lie superalgebra, together with a representation of the Clifford algebra, is promoted to a representation of the color superalgebra. 
It seems to be difficult to convert a Lie algebra to a color algebra in a similar way.  
At least some simple trials did not work well, though we do not give a detailed discussion. 

  In the second part of this paper, we showed that the boson-fermion systems can be converted to color superalgebras by the method of  \cite{aktt1,aktt2,AiSe}. The method is more restrictive than the one using the Clifford algebras. However, it is the method by which we found a $\ZZ$ graded symmetry of the L\'evy-Leblond equations. 
It is known that symmetries of some Schr\"odinger equations (e.g.  for a free particle) are generated by the Lie algebra of the Shcr\"odinger group, then by switching the degree of some  generators from even (degree 0) to odd (degree 1)  we find a spectrum generating Lie superalgebra of the quantum mechanical Hamiltonian \cite{Toppan,aikuto}.     
This leads us to the idea that  if symmetries of a differential equation are generated by a Lie superalgebra, then there exists a $ \ZZ $ graded color (super)algebra relating to the differential equation. Therefore for a given physical system we  expect a existence of sequence: from a Lie algebra to a Lie superalgebra, from a Lie superalgebra to a $\ZZ$ graded color (super)algebra, from a $\ZZ$ color (super)algebra to $ \mathbb{Z}_2^{\otimes 3}$ color (super)algebra and so on. 
We think that this would be a good way to find physical application of the color (super)algebras. 
Namely, one may find a $\ZZ$ graded color (super)algebra by investigating a system having (dynamical) supersymmetry. 

 There is no need to say the importance of representation theory of any algebraic structures when their physical and mathematical applications are concerned. We presented a vector field representation of the color superalgebra $\mathfrak{bf}$ which is a straightforward extension of the Lie theory to the $\ZZ$ setting. However, if we consider a naive extension of Verma modules to the case of $\mathfrak{bf}$, we encounter some difficulties.  
To have a lowest/highest weight representations of $\mathfrak{bf}$ we need an modified approach. 
This will be reported elsewhere. 


\subsection*{Acknowledgment}
N. A. is supported by the  grants-in-aid from JSPS (Contract No. 26400209).


\begin{thebibliography}{1}
\bibitem{rw1} V. Rittenberg and D. Wyler, \textit{Generalized Superalgebras.} Nucl. Phys. {\bf B 139} (1978), 189.

\bibitem{rw2} V. Rittenberg and D. Wyler, \textit{Sequences of $Z_2\otimes Z_2$ graded Lie algebras and superalgebras.} J. Math. Phys. {\bf 19} (1978), 2193.

\bibitem{sch} M. Scheunert, \textit{Generalized Lie algebras.} J. Math. Phys. {\bf 20} (1979), 712.

\bibitem{GrJa} H. S. Green and P. D. Jarvis, J. Math. Phys. \textit{Casimir invariants, characteristic identities, and Young diagrams for color algebras and superalgebras.} \textbf{24}  (1983), 1681.


\bibitem{sch2} M. Scheunert, \textit{Graded tensor calculus.} J. Math. Phys. \textbf{24}  (1983), 2658.

\bibitem{sch3} M. Scheunert, \textit{Casimir elements of $\epsilon$-Lie algebras.} J. Math. Phys. \textbf{24} (1983), 2671.

\bibitem{Sil} S. D. Silvestrov, \textit{On the classification of 3-dimensional coloured Lie algebras.} 
Banach Center Publications, \textbf{40} (1997), 159.

\bibitem{ChSiVO} X.-W. Chen, S. D. Silvestrov and F. Van Oystaeyen, \textit{Representations and cocycle twists of color Lie algebras.} Algebr. Represent. Theor. \textbf{9} (2006), 633.

\bibitem{CART} R. Campoamor-Stursberg and M. Rausch de Traubenberg, \textit{Color Lie algebras and Lie algebras of order $F,$.} J. Gen. Lie Thory Appl.  \textbf{3} (2009), 113.

\bibitem{SigSil} G. Sigurdsson and S. D. Silvestrov, \textit{Bosonic realizations of the colour Heisenberg Lie algebra.} J. Nonlinear Math. Phys. \textbf{13} (2006), 110.


\bibitem{AiSe} N. Aizawa and J. Segar, \textit{$\mathbb{Z}_2 \times \mathbb{Z}_2$  generalizations of ${\mathcal N} = 2$ super Schr\"odinger algebras and their representations.} J. Math. Phys. \textbf{58} (2017), 113501.

\bibitem{StoVDJ} N.I. Stoilova, J. Van der Jeugt, \textit{The $\ZZ$-graded Lie superalgebra pso(2m+1|2n) and new parastatistics representations.} arXiv:1711.02136 [math-ph].

\bibitem{lr} J. Lukierski and V. Rittenberg, \textit{Color-De Sitter and color-conformal superalgebras.} Phys. Rev. {\bf D 18}  (1978), 385.

\bibitem{vas} M. A. Vasiliev, \textit{de Sitter supergravity with positive cosmological constant and generalized
 Lie superalgebras.} Class. Quantum Grav. {\bf 2}  (1985), 645.

\bibitem{jyw} P. D. Jarvis, M. Yang and B. G. Wybourne, \textit{Generalized quasispin for supergroups.}  J. Math. Phys. {\bf 28}   (1987), 1192.

\bibitem{zhe} A. A. Zheltukhin, \textit{Para-Grassmann extension of the Neveu-Schwartz-Ramond algebra.} Theor. Math. Phys. {\bf 71}   (1987), 491 (Teor. Mat. Fiz. {\bf 71} (1987), 218).

\bibitem{Toro1} L. A. Wills-Toro, 
\textit{$(I,q)$-graded Lie algebraic extensions of the Poincar\'e algebra, constraints on $I$ and $q$.} J. Math. Phys. \textbf{36} (1995), 2085.

\bibitem{Toro2} L. A. Wills-Toro, 
\textit{Trefoil symmetries I. Clover extensions beyond Coleman-Mandula theorem.} J. Math. Phys. \textbf{42} (2001) 3915. 



\bibitem{tol2} V. N. Tolstoy,  \textit{Super-de Sitter and Alternative Super-Poincar\'e Symmetries.} In: Dobrev V. (eds) Lie Theory and Its Applications in Physics. Springer Proceedings in Mathematics \& Statistics, vol. 111, Springer, Tokyo, 2014.

\bibitem{tol} V. N. Tolstoy, \textit{Once more on parastatistics.} Phys. Part. Nucl. Lett. {\bf{11}}   (2014), 933.


\bibitem{aktt1} N. Aizawa, Z. Kuznetsova, H. Tanaka and F. Toppan, \textit{$\mathbb{Z}_2 \times \mathbb{Z}_2$-graded Lie symmetries of the L\'evy-Leblond equations.} Prog. Theor. Exp. Phys. \textbf{2016} (2016), 123A01.

\bibitem{aktt2} N. Aizawa, Z. Kuznetsova, H. Tanaka and F. Toppan, \textit{Generalized supersymmetry and L\'evy-Leblond equation.} 
in S. Duarte \textit{et al} (eds), \textit{Physical and Mathematical Aspects of Symmetries.} Springer, 2017.

\bibitem{LLE} J.-M. L\'evy-Leblond, \textit{Nonrelativistic particles and wave equations.} Comm. Math. Phys. \textbf{6} (1967), 286. 

\bibitem{Dob} V. K. Dobrev, \textit{Canonical construction of differential operators intertwining representations of real semisimple Lie groups.} Rep. Math. Phys. \textbf{25} (1988), 159. 

\bibitem{Toppan} F. Toppan, \textit{Symmetries of the Schr\"odinger equation and algebra/superalgebra duality.} 
J. Phys. Conf. Ser. \textbf{597} (2015), 012071.

\bibitem{aikuto} N. Aizawa, Z. Kuznetsova and F. Toppan, 
\textit{$\ell$-oscillators from second-order invariant PDEs of the centrally extended conformal Galilei algebras.} J. Math. Phys. \textbf{56} (2015), 031701. 

\end{thebibliography}
\end{document}